\documentclass[a4paper,11pt]{article}

\usepackage[utf8]{inputenc}
\usepackage[english]{babel}
\usepackage{amssymb,amsmath,amsthm}

\newtheorem{theorem}{Theorem}
\newtheorem{corollary}{Corollary}
\usepackage[margin=0.95in]{geometry}

\usepackage{amsthm,enumerate}
\usepackage{graphicx}
\usepackage{color,hyperref}
\usepackage{tabularx}
\usepackage{amsmath}
\usepackage{amssymb,amsfonts,textcomp}
\usepackage{latexsym}
\usepackage{graphicx} 
\usepackage{multirow}
\usepackage{microtype}
\usepackage{rotating}
\usepackage{url}
\usepackage{color}
\usepackage{comment}
\usepackage{booktabs}
\usepackage{bm}
\usepackage[]{units}
\usepackage[noend, ruled, noline]{algorithm2e}
\usepackage{tikz}
\usepackage{tcolorbox,complexity,multirow}
\usepackage{float}

\usepackage{algorithm2e}
\SetKw{Continue}{continue}

\newtheorem{myclaim}{Claim}

\newcommand{\vote}{\mathsf{vote}}

\definecolor{mypink1}{rgb}{0.858, 0.188, 0.478}

\usepackage[textsize=scriptsize,textwidth=3.3cm,shadow]{todonotes}

\newcommand{\kkcom}[1]{\todo[color=blue!20!white]{{K}: #1}}

\usetikzlibrary{arrows,decorations.markings,patterns,calc,positioning,backgrounds}
\tikzstyle{vertex} = [circle, draw=black, fill=black, scale= 0.5]
\tikzstyle{edgelabel} = [circle, fill=white, scale= 0.9]
\tikzstyle{arrow} = [line width=0.8mm,-implies,double, double distance=0.8mm]
\tikzstyle{dashedpointedline} = [line width=0.2mm,dashed,dash pattern=on 2mm off 1mm,
decoration={markings,
	mark=at position 1 with {\arrow[line width=0.2mm, scale= 1.8]{>}}
},
postaction={decorate}
] 
\tikzstyle{pointedline} = [line width=0.3mm,
decoration={markings,
	mark=at position 1 with {\arrow[line width=0.2mm, scale= 2]{>}}
},
postaction={decorate}
] 

\makeatletter
\newcommand\newtag[2]{#1\def\@currentlabel{#1}\label{#2}}
\makeatother

\title{Polynomially tractable cases\\ in the popular roommates problem} 

\author{Erika Bérczi-Kovács\\Eötvös Loránd University (ELTE), Hungary\\
MTA-ELTE Egerváry Research Group\\ on Combinatorial Optimization, Hungary 
\and
\'{A}gnes Cseh\\Hasso Plattner Institute, University of Potsdam, Germany \\
Institute of Economics,\\ Centre for Economic and Regional Studies, Hungary\and
Kata Kosztolányi\\
Eötvös Loránd University (ELTE), Hungary
\and
Attila Mályusz\\
Eötvös Loránd University (ELTE), Hungary}

\begin{document}

\maketitle

\begin{abstract}
The input of the popular roommates problem consists of a graph $G = (V, E)$ and for each vertex $v\in V$, strict preferences over the neighbors of~$v$. Matching $M$ is more popular than $M'$ if the number of vertices preferring $M$ to $M'$ is larger than the number of vertices preferring $M'$ to~$M$. A matching $M$ is called \emph{popular} if there is no matching~$M'$ that is more popular than~$M$.

Only recently Faenza et al.~\cite{FKPZ19} and Gupta et al.~\cite{GMSZ21} resolved the long-standing open question on the complexity of deciding whether a popular matching exists in a popular roommates instance and showed that the problem is $\NP$-complete. In this paper we identify a class of instances that admit a polynomial-time algorithm for the problem. We also test these theoretical findings on randomly generated instances to determine the existence probability of a popular matching in them.
\end{abstract}

\section{Introduction}

Our input is an instance $G = (V, E)$ of the stable roommates problem with strict and possibly incomplete preference lists. A matching $M$ is \emph{stable} if there is no {\em blocking pair} with respect to $M$, in other words, there is no pair of vertices $(a,b)$ such that $a$ is either unmatched or prefers $b$ to $M(a)$ ($a$'s partner in $M$) and similarly, $b$ is either unmatched or prefers $a$ to~$M(b)$. Irving~\cite{Irv85} gave a polynomial-time algorithm to decide whether $G$ admits a stable matching.

In this paper we consider {\em popularity}, a notion that is more relaxed than stability. 
For a vertex $v \in V$, $v$'s ranking over its neighbors can be extended naturally to a ranking over matchings as follows:
$v$ prefers matching $M$ to matching $M'$ if (i)~$v$ is matched in $M$ and unmatched in $M'$ or
(ii)~$v$ is matched in both and $v$ prefers $M(v)$ to $M'(v)$. 


The motivation of popularity comes from voting. Suppose an election is held between $M$ and $M'$ where each vertex casts a vote for the matching it prefers.  
We call a matching $M$ popular in the instance if it never loses an election to another matching $M'$. In voting terminology, each popular matching is a weak
{\em Condorcet winner}~\cite{Con85} 
in the corresponding voting instance.

Stable matchings are popular even in the non-bipartite case~\cite{Chu00}. Thus, if an instance admits a stable matching, then the existence of a popular matching is also guaranteed, and it can be found using Irving's algorithm for finding a stable matching~\cite{Irv85}. Some instances of the stable roommates problem do not admit a stable solution, yet they admit a popular matching, as demonstrated by Figure~\ref{fig:sr_pm}, first presented by Biró et al.~\cite{BIM10}.

\begin{figure}[ht]
	\centering
			\begin{minipage}{0.2\textwidth}
		\[
		\begin{array}{lll}
		a  : b, \ & d, \ & e \\
		b  : d, \ & a, \ & e \\
		d  : a, \ & b, \ & e \\
		e  : d, \ & b, \ & a
		\end{array}
		\]
	\end{minipage}\hspace{17mm}\begin{minipage}{0.35\textwidth}	
		\begin{tikzpicture}[scale=0.97, transform shape]
		\pgfmathsetmacro{\d}{5}
		\pgfmathsetmacro{\b}{3}
		
    \node[vertex, label=below:$a$] (A1) at (0,0) {};
	\node[vertex, label=below:$b$] (A2) at ($(A1) + (\d, 0)$) {};
	\node[vertex, label=above:$d$] (A3) at ($(A1) + (0.5*\d, 0.8*\d)$) {};
	\node[vertex, label=below:$e$] (A4) at ($(A1) + (0.5*\d, 0.3*\d)$) {};

	\draw [thick, color=gray, dashed] (A1) -- node[edgelabel, near start] {1} node[edgelabel, near end] {2} (A2);
	\draw [thick] (A2) -- node[edgelabel, near start] {1} node[edgelabel, near end] {2} (A3);
	\draw [thick, dotted] (A3) -- node[edgelabel, near start] {1} node[edgelabel, near end] {2} (A1);
	
	\draw [thick] (A1) -- node[edgelabel, near start] {3} node[edgelabel, near end] {3} (A4);
	\draw [thick, dotted] (A2) -- node[edgelabel, near start] {3} node[edgelabel, near end] {2} (A4);
	\draw [thick, color=gray, dashed] (A3) -- node[edgelabel, near start] {3} node[edgelabel, near end] {1} (A4);	
\end{tikzpicture}
\end{minipage}
	\caption{The dashed gray edges mark one of the two popular matchings $M=\{(a,b), (d,e)\}$. It is blocked by the edge $(b,d)$. The other popular matching is marked by dotted edges. The instance admits no stable matching.}
	\label{fig:sr_pm}
\end{figure}
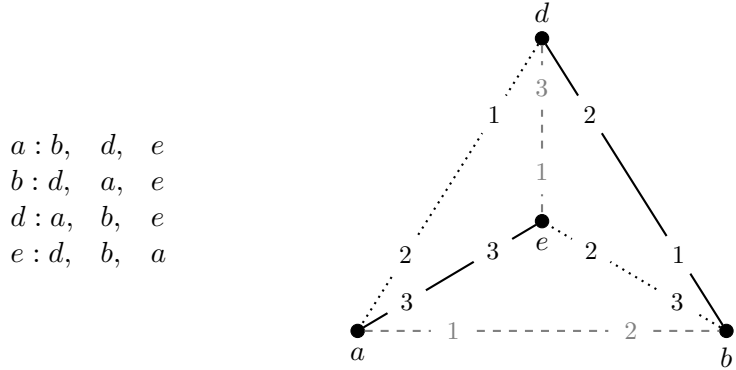

Validating whether a given matching is popular can be done in polynomial time
~\cite{BIM10}. Only recently Faenza et al.~\cite{FKPZ19} and Gupta et al.~\cite{GMSZ21} resolved the long-standing~\cite{BIM10,Cse17,HK13a,HK17,Man13} open question on the complexity of deciding whether a popular matching exists in a popular roommates instance and showed that the problem is $\NP$-complete. This hardness result remains valid even if the preference lists are complete~\cite{CK21}. However, this hardness result is only valid for complete graphs with an {\em even} number of vertices. It is known that when $n$, the number of vertices in $G$ is {\em odd}, a matching in a complete graph $G$ on $n$ vertices is popular only if it is stable~\cite{CK21}. Since Irving's algorithm can decide if $G$ admits a stable matching or not~\cite{Irv85}, the popular roommates problem in a complete graph $G$ can be efficiently solved for odd~$n$. 

\subsection{Our contribution}
In this paper, we stretch this observation even further and show that for \textit{high-degree} graphs with an \textit{odd} $n$, there is a polynomial algorithm that solves the popular roommates problem. More precisely, our algorithm runs in polynomial time for a constant $c$, where $n-c$ is the minimum degree in the graph.

Our key technical result is an algorithm to decide whether there exists a popular matching that leaves exactly a given set of vertices uncovered. We prove that popular matchings with this property can be decomposed into a stable and a popular part. We iterate through the possible popular parts and construct the fitting stable part in a carefully designed subroutine.

We also test our algorithm on randomly generated Erd\H{o}s-R\'{e}nyi graphs and conclude that on graphs with a high minimum degree, our method is significantly faster than checking the matchings in the instance for popularity. Furthermore, we also tally the instances with a popular but with no stable matching. Their number decreases for a fixed $c$ and increasing $n$, and increases for a fixed $n$ and increasing~$c$. 

\subsection{Literature review}

The notion of popularity was first introduced for bipartite graphs in 1975 by G\"ardenfors---popular matchings always exist in bipartite graphs since stable matchings always exist here~\cite{GS62} and every stable matching is popular~\cite{Gar75}. The proof that every stable matching is popular holds in non-bipartite graphs as well~\cite{Chu00}; in fact, it is easy to show that every stable matching is a min-size popular matching in both settings~\cite{HK11}. 





Huang and Kavitha~\cite{HK17} showed that the polytope of popular fractional matchings is half-integral in the non-bipartite case. This means that one can compute a maximum weight popular half-integral matching in polynomial time. They also showed that the problem of computing an integral maximum weight popular matching in a non-bipartite instance is $\NP$-hard. Finding a maximum cardinality popular matching in instances admitting a popular matching has also been shown to be $\NP$-hard~\cite{Kav18}. The breakthrough regarding the complexity of the popular roommates problem came after these works: two independent $\NP$-completeness proofs, by Faenza et al.~\cite{FKPZ19} and Gupta et al.~\cite{GMSZ21}, were published in 2019.

Most positive results in the setting are also very recent. The only known tractable subclasses of popular matchings are the class of stable matchings and the class of so-called `strongly dominant matchings', which is a subclass of max-size popular matchings~\cite{FKPZ19}. If the underlying graph $G$ has a bounded treewidth, then the min-cost popular matching problem can be solved in polynomial time~\cite{FKPZ19}. Kavitha provided a simply exponential time algorithm for the popular roommates problem~\cite{Kav19}. Interestingly, its running time was slightly decreased very recently by Palmer and P\'alv\"olgyi~\cite{PP20}, who improved the upper bound on the number of stable matchings in bipartite matching instances. Huang and Kavitha~\cite{HK13a} proved that each roommates instance admits a matching that has the approximability measure called `unpopularity factor' of $O(\log n)$. 

The existence probability of a stable matching in randomly generated roommates instances has been an actively researched topic for decades, which lead to both theoretical~\cite{Irv85,Pit93,PI94} and experimental findings~\cite{Pro14,Mer15,EFM+20}. Until now, randomly generated popular matching instances have only been studied in the context of bipartite graphs, both with one-sided~\cite{Mah06,IW10,RI18} and two-sided preferences~\cite{AIKM07}.

\section{Preliminaries}

Consider a graph $G = (V,E)$ with $n$ vertices and $m$ edges. Let $N(u)$ denote the vertices adjacent to $u \in V$ in $G$, and let $N(U)$ for $U \subseteq V$ be the set of vertices that are adjacent to at least one vertex in~$U$. Furthermore, let $G[U]$ be the graph spanned by the vertex set~$U$.

Each vertex ranks all adjacent vertices in a strict order of preference. We write $v \succ_u w$ if $u$\ prefers\ $v$\ to\ $w$. A \textit{matching} $M \subseteq E$ is a set of edges such that each vertex is incident to at most one edge in~$M$. For convenience, we denote $u$'s partner in $M$ by $M(u)$, and write $M(u) = u$ if $u$ is unmatched in~$M$. For a set $X \subseteq V$, $M(X) = \bigcup_{v\in X}M(v)$. We assume that $v \succ_u u$ for all $v \in N(u)$. In other words, each vertex prefers being matched along any of its edges to staying unmatched.

We now introduce the edge labeling technique that is standard in the field~\cite{HK11}. Let $M$ be a matching in $G$. For each edge $(u,v) \notin M$, we define the vote of vertex $u$ between the edge $(u,v)$ and $M$ as follows: 
\begin{equation*}
\vote_u(v,M) = \begin{cases} +   & \text{if\ $v \succ_u M(u)$};\\
	                     - &  \text{if\ $M(u) \succ_u v$.}			
\end{cases}
\end{equation*}

The next step is to label every edge $(u,v)$ that does not belong to $M$ by the pair $(\vote_u(v,M),\vote_v(u,M))$. Thus every non-matching edge has a label in $\{(\pm,\pm)\}$. Note that an edge is labeled $(+,+)$ if and only if it blocks~$M$. Let $G_M$ be the subgraph of $G$ obtained by deleting edges labeled $(-,-)$ from $G$. The following theorem characterizes popular matchings in~$G$. 

\begin{theorem}[Huang and Kavitha~\cite{HK11}]
  \label{thm:char-popular}
$M$ is popular in $G$ if and only if $G_M$ does not contain any of the following with respect to~$M$:
	\begin{enumerate}
		\item an alternating cycle with a $(+,+)$ edge;
		\item an alternating path with two disjoint $(+,+)$ edges;
		\item an alternating path with a $(+,+)$ edge and an unmatched vertex as an end vertex.
	\end{enumerate}
\end{theorem}

Using the above characterization, it can be easily checked whether a given matching $M$ is popular or not~\cite{KMN11}: $M$ is popular if and only if it is a maximum weight perfect matching in the weighted graph we get by first adding a loop to each vertex and then assigning weight -1 to those loops that belong to vertices covered by $M$, 2 to all $(+,+)$ edges, -2 to all $(-,-)$ edges, and 0 to all other edges and loops. However, this only settles the complexity of verification.

\textbf{Structure of the paper.} We first explain the high-level idea of our algorithm in Section~\ref{sec:alg}. Its proof of correctness and the exact implementation details are then given in Section~\ref{sec:proofs}. Section~\ref{sec:ex} contains a detailed example of the execution. In Section~\ref{sec:exp} we analyze the performance of our algorithm on randomly generated graphs. We conclude and pose open questions in Section~\ref{sec:concl}. Two full example runs are to be found in the appendix.

\section{Our algorithm}
\label{sec:alg}

As we have already observed in the Introduction, if the instance admits a stable matching, then it is popular as well, so the existence question is only interesting for instances that do not admit any stable matching. From now on, we restrict our attention to these instances.

We are given an instance of the popular roommates problem that does not admit a stable matching. Let $U$ be
a non-empty set of vertices. 
We now present an algorithm with which it can be checked whether there is a popular matching in $G$ that leaves exactly $U$ uncovered.


To a given $U$, let $Z=V\setminus N(U)\setminus U$ denote the set of those vertices that are not in $U$ and also not adjacent to \emph{any} vertex in~$U$. It is clear that the set $Z$ is well-defined for each set~$U$. In our algorithm (see Algorithm~\ref{alg: check U} for a pseudocode), we build a set $\mathcal{P}_Z$ of initial matchings and test for each $P_Z \in \mathcal{P}_Z$ whether it can be completed to a popular matching that leaves exactly $U$ uncovered. The construction of $\mathcal{P}_Z$ is simple: we list all matchings in which each edge has at least one end vertex in $Z$ 
and also cover each vertex in~$Z$. Then, we examine each such $P_Z \in \mathcal{P}_Z$ in three steps.
\begin{enumerate}
    \item If $P_Z$ is not popular in $G'=G[Z\cup P_Z(Z) \cup U]$, then we stop testing this $P_Z$ and conclude that $P_Z$ cannot be extended to a desired popular matching.
    \item  If there is no path $v-P_Z(v)-P_Z(w)-w$ with $(P_Z(v),P_Z(w))$ blocking $P_Z$, then we output the same message and stop examining~$P_Z$.
    \item Otherwise, there is at least one alternating path traversing through a blocking edge in~$G'$. 
    In Claim~\ref{extension} we show that matching $P_Z$ can be extended to a popular matching in $G$ with edge set $S$ leaving exactly $U$ uncovered if and only if $S$ is a stable complete matching in a transformed graph. Later, in Section~\ref{sec:imp} we will show how the existence of a desired stable matching can be checked in polynomial time.
\end{enumerate}

\begin{algorithm}
\DontPrintSemicolon 
\KwIn{A popular matching problem instance that admits no stable matching and a non-empty set of vertices $U\subseteq V$.}
\KwOut{A popular matching covering exactly $V\setminus U$ or a 'NO' answer.}
\nl $Z:=$ set of vertices in $V\setminus U$ that are not adjacent to any vertex in $U$\; 
\nl $\mathcal{P}_Z:=$ set of matchings $P_Z$ covering $Z$ such that for every edge $(p,q)\in P_Z$: $Z\cap \{p,q\} \neq \emptyset$\;
\nl \For {$P_Z \in \mathcal{P}_Z$} {
    \nl \If {$P_Z$ is not popular in $G'=G[Z\cup P_Z(Z) \cup U]$} {\Continue\;}
    \nl \If {there is no blocking edge for $P_Z$ in $G'$} {\Continue\;}
    \nl \If {a complete stable matching $S$ exists on $V\setminus V'$ fulfilling the conditions in Claim~\ref{extension} }{\Return{$P_Z\cup S$} \;}
    }
\nl \Return{\text{`NO'}}\;
\caption{Subroutine checking if set $U\neq \varnothing$ can be the set of uncovered vertices in a popular matching}\label{alg: check U}
\end{algorithm}

\section{Proof of correctness}
\label{sec:proofs}

In Sections~\ref{sec:1t}-\ref{sec:3t} we prove the necessity of our three tests in lines~4-6 in our pseudocode. 
Then in Section~\ref{sec:imp} we elaborate on the implementation of the last one of these. Finally in Section~\ref{sec:corr} we complete the proof of correctness and provide a running time analysis.

\subsection{First test}
\label{sec:1t}
Our first claim shows the necessity of the first step in line~4. 
More precisely, we show that if a matching $P_Z$ is not popular in $G'$, then it cannot be extended to a popular matching in $G$, and thus we can proceed to testing the next candidate for~$P_Z$.

\begin{myclaim}\label{P_Z not pop}
    If the constructed matching $P_Z$ is not popular in $G'$, then no matching $P$ containing $P_Z$ is popular in~$G$.
\end{myclaim}
\begin{proof}
    The augmenting path or cycle $P_Z$ admits in $G'_{P_Z}$ also exists in $G_P$ if $P_Z \subseteq P$, since $G'_{P_Z}$ is a subgraph of~$G_P$.
\end{proof}

\subsection{Second test}
\label{sec:2t}
Next we turn to the second test in line~5. The key observation is that edges blocking $P$ cannot occur anywhere in the graph, but only between vertices that $P$ matches to~$Z$. This observation is a corollary of the following claim.

\begin{myclaim}\label{block-component}
    Let $P$ be a popular matching that leaves exactly $U$ uncovered. 
    If edge $(p,q) \in P$ can be reached on an alternating path in the graph $G_P$ from an edge that blocks $P$ such that $q$ is the last vertex on that path, then $q$ is in~$Z$.
\end{myclaim}
\begin{proof}
    We indirectly suppose that there is an edge $(p,q) \in P$ that can be reached on an alternating path in $G_P$ from a blocking edge such that $q$ is the last vertex on that path, yet $q \notin Z$. If $q \notin Z$ then $q$ must be in $N(U)$, since no vertex in $U$ is matched in~$P$. Now we take the alternating path in $G_P$ that traverses through $(p,q)$. We lengthen this path past $q$ toward $U$ along the $(+,-)$ edge that connects $q$ to some $u \in U$.
    With this new path we have shown that to $P$, there is an alternating path that traverses a blocking edge and ends in an unmatched vertex, which contradicts the assumption on the popularity of~$P$.
\end{proof}

\begin{corollary}\label{blocking}
Let $P$ be a popular matching that leaves exactly $U$ uncovered. Each edge blocking $P$ connects two vertices that $P$ matches to vertices in~$Z$.
\end{corollary}

\subsection{Third test}
\label{sec:3t}
For better readability, we introduce the notation $V' = Z\cup P_Z(Z) \cup U$ for the vertex set of graph~$G'$. Notice that $V\setminus V' = N(U) \setminus P_Z(Z)$, which is the set of vertices that are adjacent to at least one vertex in $U$ and not covered by~$P_Z$. Let us denote by $D$ the set of `dangerous' vertices. These are vertices 
that can be reached on an alternating path from a blocking edge in $G'_{P_Z}$ 
such that counting from the blocking edge, $z$ is the further end of the matching edge $(P(z),z)$ on the path. Note that from Claim~\ref{block-component} we get that $D\subseteq Z$.

Dangerous vertices play a crucial role in our third test in line~6. These vertices can be reached on an alternating path from a blocking edge, which path must be discontinued in $G_P$ before it reaches~$U$. Our next claim guarantees that we extend $P_Z$ to a $P$ that fulfills this criterion, provided it is possible for the given $P_Z$ at all.

\begin{myclaim}\label{extension}
    Assume that $P_Z$ is popular but not stable in~$G'$. 
    In $G$ there is a popular matching $P$ that contains $P_Z$ and leaves exactly $U$ uncovered if and only if $P \setminus P_Z$ is a stable matching $S$ in the graph spanned by the vertices $V\setminus V'$, 
    covers exactly the vertices in $V\setminus V'$, and induces the following edge labeling with respect to $P= P_Z\cup S$.

    \begin{enumerate}
        \item $(+,-)$ for edges incident to $U$, with a  $'+'$ at the vertex in $U$,
        \item $(-,-)$ for edges between $D$ and $V\setminus V'$,
        \item $(+,-)$ for an edge $(v',x)$ 
        where $v' \in V'\setminus (D \cup U)$,  $x\in V\setminus V'$, and $v'$ prefers $x$ to $P(v')$, with a $'+'$ at~$v'$. 
    \end{enumerate}
\end{myclaim}

\begin{proof}
Let us first assume that in $G$ there is a popular matching $P$ that contains $P_Z$ and leaves exactly $U$ uncovered. Since $V\setminus V' = N(U) \setminus P_Z(Z)$ and $P_Z$ must cover all vertices in $Z$ by construction, $S = P \setminus P_Z$ must cover exactly the vertices in $V\setminus V'$. We will first show that $S$ is a stable matching, and then prove that $P$ induces an edge labeling satisfying the three points above.

The stability of $S$ is easy to show. Without loss of generality we can assume that $P_Z$ does not cover the entire set $Z \cup N(U)$, because otherwise $S$ is defined on an empty graph. Assume now that edge $(p,q)$ blocks $S$, which we have just shown to cover exactly the vertices in $N(U) \setminus P_Z(Z)$. Then there exists a vertex $u \in U$ such that $u - P(p) - p - q - P(q)$ is an augmenting path in $G_P$, which contradicts the popularity of~$P$.

We now check the edge labeling property in each point separately.
    \begin{enumerate}
        \item Edges incident to $U$ trivially have a $'+'$ at the vertex in~$U$. The other label must be a $'-'$, otherwise we found a blocking edge at an unmatched vertex, which contradicts the popularity of~$P$.
        \item Assume indirectly that there are vertices $v\in V\setminus V'$ and $z\in D$ such that $(v,z)$ is not labelled $(-,-)$ with respect to~$P$. By the definition of $V'$, there exists a vertex $u \in U$ such that $u - P(v) - v - z$ is an augmenting path in $G_P$, furthermore by the definition of $D$, a blocking edge is reachable from $z$ via an alternating path in~$G_P$. This latter path is disjoint from $(v,P(v))$, because a blocking edge can only occur in the subgraph spanned by $Z \cup P_Z(Z)$ (see Corollary~\ref{blocking}) and from $z \in D$ it can only be reached via an alternating path that also runs in this subgraph, while both $v$ and $P(v)$ are in $V\setminus V'$. 
        The concatenation of these two paths is therefore an augmenting path to $P$ in $G_P$, which contradicts the popularity of~$P$.
        \item Since $v'$ prefers $x$ to $P(v')$, a $'+'$ at $v'$ is guaranteed. Assume indirectly that there are vertices $v' \in V'\setminus (D \cup U)$ and $x \in V\setminus V'$ such that $(v',x)$ blocks~$P$. Then by the definition of $V\setminus V'$, there is a vertex $u \in U$ such that $u - P(x) - x - v' - P(v')$ is an augmenting path in $G_P$, which contradicts the popularity of~$P$.
    \end{enumerate}

We now turn to proving the opposite direction. Assume that $S$ is stable in the graph spanned by the vertices $V\setminus V'$, it covers exactly $V\setminus V'$, and the three points on edge labeling are fulfilled. Since $(Z \cup P_Z(Z)) \cup (V\setminus V') = V \setminus U$, matching $P=P_Z\cup S$ leaves exactly $U$ uncovered. We will next utilize Theorem~\ref{thm:char-popular}, the characterization of popular matchings by Huang and Kavitha~\cite{HK13}, to show that $P$ is popular in~$G$.
    \begin{enumerate}
        \item $G_P$ admits no alternating cycle that contains a blocking edge.\\
        Neither $P_Z$, nor $S$ creates an alternating cycle with a blocking edge in the subgraphs $G_{P_Z}$ on the vertex set $Z \cup P_Z(Z)$ and $G_S$ on the vertex set $V\setminus V'$, respectively. Therefore, for such an alternating cycle in $G_P$, edges between the two sets of vertices must connect paths to form a cycle. However, our three points enforce that no blocking edge leaves $V\setminus V'$. Therefore, the blocking edge must be in~$G_{P_Z}$. These blocking edges can only be reached from dangerous vertices on an alternating path in $G_P$, and due to point~2 above, those paths are all disrupted by a $(-,-)$ edge between $V'$ and $V\setminus V'$.
        \item $G_P$ admits no alternating path that contains at least two blocking edges.\\
        As we have argued above, all blocking edges in $G_P$ must also block~$P_Z$. However, an entire alternating path with two blocking edges cannot run in $V'$, because it would contradict the popularity of~$P_Z$. For an alternating path containing a blocking edge, leaving $V'$ is also not an option, because of point~2 above.
        \item $G_P$ admits no alternating path that contains a blocking edge and starts at a vertex in~$U$.\\
        Similarly as above, the blocking edge must also block~$P_Z$, and the alternating path would have to leave $V'$, which is impossible due to point~2 above.
\end{enumerate}
With this we have shown that all three points are fulfilled by $P = P_Z \cup S$.
\end{proof}

\subsection{Implementation of the third test}
\label{sec:imp}
We now sketch the main idea on how to enforce the three points on edge labeling in Claim~\ref{extension}. For a pseudocode, please consult Algorithm~\ref{alg:delete}.

The two $'+'$ signs among the $3\times2$ edge labels are fulfilled trivially. Three of the 4 $'-'$ signs impose a requirement on vertices in $V\setminus V'$, while the fourth one imposes a requirement on vertices in~$D$. This latter one can be found in point 2, which requires $'-'$ votes for vertices in $D$ on edges between $D$ and $V\setminus V'$. Hence if a vertex $z\in D$ votes $'+'$ for a neighbor $x\in V\setminus V'$ then the stable matching $S$ described in Claim~\ref{extension} cannot exist. We test this property first, as it does not depend on how $P_Z$ is completed to~$P$. This constitutes the first phase of our subroutine.

For the completion of $P_Z$, we ensure the remaining three $'-'$ signs in three edge deletion rounds of the second phase as follows. For a vertex $x$ in $V\setminus V'$, each point determines a list of edges incident to $x$ such that the vote of $x$ should be a  $'-'$ for that edge with respect to~$P$. Each of these requirements imposes an upper bound on the rank of possible partners for each vertex in $V\setminus V'$. Our key step is to delete all the edges that are ranked worse than any of these bounds. If there is a stable matching in the remaining edge set covering every vertex in $V\setminus V'$, then this is a suitable matching $S$, because none of the deleted edges would have blocked it, as every deleted edge has at least one vertex where the bound, and thus the matching partner is higher in the preferences than that edge.

\begin{algorithm}
\DontPrintSemicolon 
\KwIn{A popular matching problem instance and an initial matching $P_Z$ that has passed the first two tests.}
\KwOut{A complete stable matching $S$ on $V\setminus V'$  or a 'NO' answer.}
\nl \For {$z \in D$ and $x \in N(z)\cap (V\setminus V')$} {
\nl \If {$x \succ_z P_Z(z)$}{\Return{\text{`NO'}}\;}}
\nl \For {$u \in U$} {
\nl \For {$x \in N(u) \cap (V\setminus V')$} {
    \nl delete all $(x,y)$ edges in $G[V\setminus V']$ such that $u \succ_x y$; 
    }
    }
\nl \For {$z \in D$} {
\nl \For {$x \in N(z)\cap (V\setminus V')$} {
    \nl delete all $(x,y)$ edges in $G[V\setminus V']$ such that $z \succ_x y$; 
    } 
    } 
\nl \For {$v \in V'\setminus (D\cup U)$} {
\nl \For {$x \in N(v)\cap (V\setminus V')$ such that $x\succ_v P_Z(v)$} {
    \nl delete all $(x,y)$ edges in $G[V\setminus V']$ such that $v \succ_x y$; 
    }
    }
\nl \If {a complete stable matching $S$ exists in the remaining graph $G[V\setminus V']$}{\Return{$S$} \;}

\nl \Return{\text{`NO'}}\;
\caption{Subroutine checking if a complete stable matching $S$ exists on $V\setminus V'$ fulfilling the conditions in Claim~\ref{extension}}
\label{alg:delete}
\end{algorithm}

The correctness of Algorithm~\ref{alg:delete} is proven by the following claim.

\begin{myclaim}
Assume that we are given a matching $P_Z$ that has passed the first two tests of Algorithm~\ref{alg: check U} and the first phase of the third test. A matching $S$ in the graph $G[V\setminus V']$ is stable, covers all vertices in $V\setminus V'$, and induces the three-point edge-labeling criteria phrased in Claim~\ref{extension} if and only if it is a stable matching that covers all vertices in $V\setminus V'$ in the graph we get after deleting each edge $(x,y) \in (V\setminus V') \times (V\setminus V')$ that fulfills any of the following points.
\begin{enumerate}
    \item $\exists u \in U$: $u \succ_x y$
    \item $\exists z \in D$: $z\succ_x y$
    \item $\exists v \in V'\setminus (D\cup U)$: $x\succ_v P_Z(v)$ and $v\succ_x y$
\end{enumerate}
\label{cl:deletion}
\end{myclaim}
\begin{proof}
The condition on the set of covered vertices in the transformed graph is clearly necessary and sufficient, while it is clear that the matching in the new instance must be stable to ensure stability in the instance before edge deletions. We now revisit the three points in Claim~\ref{extension} one-by-one and translate them into equivalent edge deletions. Claim~\ref{extension} states that for each vertex $x \in V\setminus V'$, $P$ must fulfill the following criteria:
\begin{enumerate}
    \item $(u,x)$ is a $(+,-)$ edge for $\forall u \in U$,
    \item $(z,x)$ is a $(-,-)$ edge for $\forall z \in D$,
    \item $(v,x)$ is a $(+,-)$ edge for $\forall v \in V'\setminus (D\cup U)$ such that $x\succ_v P(v)$.
\end{enumerate}


Notice that the $'+'$ signs are unavoidable, while the first phase of Algorithm~\ref{alg:delete} guarantees the first $'-'$ sign in point~2. Thus we only need to make sure that the conditions imposed by the remaining three $'-'$ signs are fulfilled. These conditions state that each $x \in V\setminus V'$ is matched to a partner who is ranked better than each of 1) $u \in U$, 2) $z \in D$, and 3) $v \in V'\setminus (D\cup U)$ such that $x\succ_v P_Z(v)$. 
These conditions enforce that all edges that are ranked worse by $x$ than any of these three bounds need to be deleted to guarantee the three $'-'$ signs. To be more precise, an edge $(x,y) \in (V\setminus V') \times (V\setminus V')$ needs to be deleted if it fulfills any of the three points listed in Claim~\ref{cl:deletion}.


It is easy to see that if these three types of edge deletions are executed, then the edge labeling from Claim~\ref{extension} is guaranteed. All that remains is to show is that stable matchings in the new instance are also stable in the original one. If such a matching $M$ is blocked by an edge $(x,y)$, then $(x,y)$ must have been removed in the edge deletion round. However, due to the cover constraint on $V \setminus V'$, $M$ must match at least one of $x,y$ along an edge that is ranked better than~$(x,y)$. 
\end{proof}

\subsection{Correctness and running time}
\label{sec:corr}
\begin{theorem}
For a vertex set $U$ and $Z=V\setminus N(U)\setminus U$, a popular but not stable matching that leaves exactly the vertices in $U$ uncovered or a proof for its non-existence is outputted by Algorithm~\ref{alg: check U} in $O(n^{|Z|})(m + |Z|^3)$ time.
\end{theorem}
\begin{proof}
We first prove that if Algorithm~\ref{alg: check U} terminates with a matching, then it is popular and leaves exactly the vertices in $U$ uncovered. Let $P=P_Z\cup S$ be the outputted matching. By definition, $P_Z$ covers the vertices in $Z \cup P_Z(Z)$, while $S$ covers the vertices in $V\setminus V' = N(U) \setminus P_Z(Z)$, which is altogether $Z \cup N(U) = V \setminus U$. The popularity follows from Claim~\ref{extension}.

For the other direction we assume that $P$ is a popular, but unstable matching in $G$, and that it leaves exactly $U$ uncovered. Algorithm~\ref{alg: check U} tests all $P_Z$ matchings that cover $Z$ and are inclusion-wise minimal with respect to this property. We know from Corollary~\ref{blocking} that $P$ is blocked by an edge that connects two vertices in $Z \cup P_Z(Z)$ for such a $P_Z$ matching. This $P_Z$ must therefore be generated in line~2, and pass the tests in lines~4 and~5. Furthermore, Claim~\ref{extension} shows that $S = P \setminus P_Z$ exists, which then guarantees that line~6 is also passed for this particular~$P_Z$. Note that more than one suitable stable matchings might exist---our algorithm only outputs one stable matching that extends $P_Z$ to a popular matching, which might be different from $P$ itself.

\textbf{Running time.} We assume that the input is given in a form of a vertex set marking $U$ and the strictly ordered list of edges at each vertex, in which list we assume comparison can be done in constant time. The lists also provide an edge list for the whole graph. The running time of Algorithm~\ref{alg: check U} is determined by the following factors.
\begin{enumerate}
    \item Line~1: determine $Z$.\\
    This takes $O(m)$, because it can be done by reading the list of edges once and deleting a vertex from $V$ if it is in $U$, or adjacent to any vertex in $U$.
    \item Line~2: determine $\mathcal{P}_Z$ to~$Z$.\\
    Choosing a partner for each vertex in $Z$ can be done in $O\left( \binom{n}{|Z|} \right)$ ways. 
    Each of the feasible $P_Z$ matchings constructed in this step needs to be investigated separately (line~3). 
    \item Line~4: check the popularity of $P_Z$ in~$G'$.\\
    We need to construct the edge labeling with respect to $P_Z$ in $O(m)$ time and then check whether $P_Z$ is indeed a maximum weight matching in the instance. 
    This can be done in $O(|Z|^3)$ time ~\cite[Theorem~26.2]{Sch03}.
    \item Line~5: check for blocking edges in $G'$.\\
    The labeling in the previous step already locates all these blocking edges.
    \item Line~6: checking the conditions in Claim~\ref{extension}.
    \begin{itemize}
        \item Determining $D$ can be done by growing alternating trees starting from each blocking edge in~$G'_{P_Z}$. There are at most $O(|Z|)$ end vertices of blocking edges, and even checking all paths starting from them can be done in $O(|Z|^3)$ time ~\cite[24.2a]{Sch03}.
        \item Testing in the first phase of Algorithm~\ref{alg:delete} can be done in constant time. Also, based on our three points in Claim~\ref{cl:deletion}, for each edge it can be decided in constant time whether it should be deleted or not.
        \item Irving's algorithm~\cite{Irv85} can find a stable matching or a proof for its non-existence in the transformed graph in $O(m)$ time. If the found stable matching does not cover all vertices in $V \setminus V'$, then no stable in this instance does, due to the Rural Hospitals Theorem~\cite[Theorem 4.5.2]{GI89}.
    \end{itemize}
\end{enumerate}
In total, this yields the following running time.
$$
O\left( m + \binom{n}{|Z|} \cdot (m+|Z|^3+|Z|^3+m+m) \right) \longrightarrow O\left( n^{|Z|} \cdot (m+|Z|^3)\right)$$
\end{proof}

Our algorithm can test on any instance of the popular roommates problem whether there is a popular matching that leaves a given non-empty vertex set $U$ uncovered. In graphs on an odd number of vertices, all matchings leave a non-empty set of vertices uncovered. Therefore, if $n$ is odd, we can iterate through all possible $U$ sets and derive whether a popular matching exists in the instance at all. For graphs with an even $n$, our algorithm is only able to decide whether a non-perfect popular matching exists. This is not surprising, because deciding whether a (perfect) popular matching exists in a complete graph with an even $n$ is $\NP$-complete~\cite{CK21}, while our algorithm runs in polynomial time if $|U|$ is small and the minimum degree in the graph is large.

On the positive side, for an odd $n$ and $\deg(v) \geq n-c$ 
for all $v \in V$ and some constant $c$, from the definition of $Z$ follows that $|Z| < c$. 
Since no vertex $u\in U$ is adjacent to a vertex in $U\cup Z$, we have that $|U\cup Z|\leq c$. Thus, the running time for checking the existence of a popular matching is polynomial for a constant $c$, as the following calculation demonstrates.
$$O\left( \sum_{i=1}^{c}\binom{n}{i} \cdot n^{c-i} \cdot (m+(c-i)^3)\right) \rightarrow O\left( n^{c}\cdot (m+c^3) \right)$$
By iterating through all possible sets of uncovered vertices in the order of $|U|$, our algorithm is able to find a maximum size popular as well in graphs with an odd $n$ and $\deg(v) \geq n-c$  for all $v \in V$.


\section{Example}
\label{sec:ex}

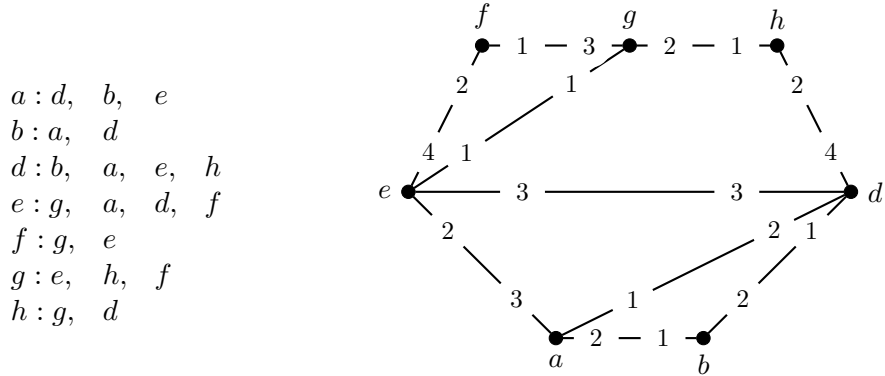
\begin{figure}[htb]
	\centering
		\begin{minipage}{0.2\textwidth}
		\[
		\begin{array}{llll}
		a  : d, \ & b, \ & e \\
		b  : a, \ & d \\
		d  : b, \ & a, \ & e, \ & h \\
		e  : g, \ & a, \ & d, \ & f \\
		f  : g, \ & e \\
		g  : e, \ & h, \ & f \\
		h  : g, \ & d
		\end{array}
		\]
	\end{minipage}\hspace{17mm}\begin{minipage}{0.35\textwidth}	
		\begin{tikzpicture}[scale=0.97, transform shape]
		\pgfmathsetmacro{\d}{2}
		\pgfmathsetmacro{\b}{3}
		
    \node[vertex, label=below:$a$] (A1) at (0,0) {};
	\node[vertex, label=below:$b$] (A2) at ($(A1) + (\d, 0)$) {};
	\node[vertex, label=right:$d$] (A3) at ($(A2) + (\d, \d)$) {};
	\node[vertex, label=left:$e$] (A4) at ($(A1) + (-\d, \d)$) {};
	\node[vertex, label=above:$f$] (A5)
	at ($(A4) + (0.5*\d, \d)$) {};
	\node[vertex, label=above:$g$] (A6)
	at ($(A5) + (\d, 0.0*\d)$) {};
	\node[vertex, label=above:$h$] (A7)
	at ($(A6) + (\d, -0.0*\d)$) {};
	
	\draw [thick] (A1) -- node[edgelabel, near start] {2} node[edgelabel, near end] {1} (A2);
	\draw [thick] (A2) -- node[edgelabel, near start] {2} node[edgelabel, near end] {1} (A3);
	\draw [thick] (A3) -- node[edgelabel, near start] {2} node[edgelabel, near end] {1} (A1);
	\draw [thick] (A1) -- node[edgelabel, near start] {3} node[edgelabel, near end] {2} (A4);
	\draw [thick] (A3) -- node[edgelabel, near start] {3} node[edgelabel, near end] {3} (A4);
	\draw [thick] (A3) -- node[edgelabel, near start] {4} node[edgelabel, near end] {2} (A7);
	\draw [thick] (A4) -- node[edgelabel, near start] {4} node[edgelabel, near end] {2} (A5);
	\draw [thick] (A4) -- node[edgelabel, near start] {1} node[edgelabel, near end] {1} (A6);
	\draw [thick] (A5) -- node[edgelabel, near start] {1} node[edgelabel, near end] {3} (A6);
	\draw [thick] (A6) -- node[edgelabel, near start] {2} node[edgelabel, near end] {1} (A7);
\end{tikzpicture}
\end{minipage}
\caption{An example instance with no stable matching and no popular matching.}
\label{fig:ex}
\end{figure}

We demonstrate our algorithm on a chosen $P_Z$ in the instance depicted in Figure~\ref{fig:ex}. Here we only discuss the case $U=\{b\}$, $P_Z= \{(e,f),(g,h)\}$. A fully detailed example that checks the existence of a popular matching in the same instance can be found in the appendix.

For this $P_Z$, the graph $G'=G[Z\cup P_Z(Z) \cup U]$ is defined by vertices $V'=\{b,e,f,g,h\}$ and edges $\{(e,f),(e,g),(f,g),(g,h)\}$. 
The first test in Algorithm~\ref{alg: check U} checks whether $P_Z$ is popular in~$G'$. As only one edge in $G'$, namely $(e,g)$, blocks $P_Z$ and from this edge there is no alternating path or cycle in $G'_{P_Z}$, based on the characterization in Theorem~\ref{thm:char-popular} we can conclude that $P_Z$ is popular in~$G'$. We remark that in practice, the fulfillment of the characterization is checked such that the three points are converted into a weight requirement in a weighted matching instance defined specifically for~$P_Z$. The second test in Algorithm~\ref{alg: check U} only requires the existence of a blocking edge for $P_Z$ in $G'$, which is granted by $(e,g)$. The first two tests are thus passed by~$P_Z$.

Consider now the third test in Algorithm~\ref{alg: check U}. For this test we build the set of remaining vertices $V\setminus V'=\{a,d\}$ and the set of dangerous vertices $D=\{f,h\}$. Vertex $f$ has no neighbor in $V\setminus V'$, $h$ is adjacent to $d$, and $P_Z(h)\succ_h d$, so the stable matching $(a,d)$ passes the first phase of Algorithm~\ref{alg:delete} in lines~1-2.  
However, in line~5 of Algorithm~\ref{alg:delete} this edge is deleted, because $d\in N(b)\cap(V\setminus V')$ and $b\succ_d a$. Therefore the remaining graph on $V\setminus V'$ is the empty graph on the two vertices $a$ and $d$, thus there is no complete matching that covers these two vertices after the edge deletion steps, so Algorithm~\ref{alg:delete} returns the answer 'NO'.

We conclude that this $P_Z$ fails the third test in Algorithm~\ref{alg: check U} because on $V\setminus V'$ there is no complete stable matching after the edge deletions.

\section{Computational study}
\label{sec:exp}

We tested our algorithm on various randomly generated instances. For each combination of $n \in \left\{7,9,11\right\}$ and $c \in \left\{3,4,5 \right\}$, 1 million graphs of minimum degree $n-c$ were generated using the Erd\H{o}s-R\'{e}nyi model~\cite{ER60}. Each edge was added with probability $p=0.8$,  
and only those graphs were kept that met the bound $n-c$ on the minimum degree and indeed had at least one vertex of degree exactly $n-c$. The preference list of each vertex was then generated by choosing uniformly at random among all permutations of its edges.

We first ran Irving's algorithm~\cite{Irv85} on each generated instance to test whether it admits a stable matching, and for those with no stable matching, we also ran our algorithm for all odd-cardinality independent vertex sets $U$ with $|U| \leq c$ to test whether there is a popular matching leaving exactly $U$ uncovered. These tests delivered the number of instances with no stable matching, and the number of instances with a popular, but no stable matching among the 1 million tested instances. Table~\ref{ta:exp} contains these numbers for all tested $(n,c)$ combinations.

Focusing on one row of the table, one can observe that for a fixed $c$ and increasing $n$, the number of instances with no stable matching increases, while the number of instances among them that admit a popular matching decreases. Fixing $n$ and increasing $c$ leads to less instances with no stable matching, and more instances among them that admit a popular matching. In general we can observe that for randomly graphs with a high minimum degree, very few instances occur that admit a popular matching but no stable matching. The monotonicity in the existence probability of a stable matching is consistent with results measured in earlier studies~\cite{Irv85,Pro14}.

The experiments were run on a standard desktop computer powered by Intel i5-3470 CPU running at 3.6GHz and 20 GiB of RAM. 
We computed pairings using the LEMON Graph Library \cite{DEZSO201123} implementation of the maximum weight perfect matching algorithm, which is based on the blossom algorithm of Edmonds~\cite{edmonds1965paths}.

\begin{table}[tb]
\centering
\begin{tabular}{|l|r|r|r|}
\hline
 & $n=7$ & $n=9$ & $n=11$ \\ \hline
\multirow{2}{*}{$c=3$} &  384678 & 508843 & 598525 \\
                        &  146  &  32 &10
                        \\ \hline
\multirow{2}{*}{$c=4$} &  298860 & 448599 & 553813 \\
                        &  1415  & 216  & 38\\ \hline
\multirow{2}{*}{$c=5$} &  211911 & 384468 & 506958 \\
                        &  8195  & 914  & 138\\ \hline
                        
\end{tabular}
\caption{The tested $n$ and $c$ values are organized in the different columns and rows. Each cell contains the number of instances with no stable matching as the top entry and the number of instances with a popular, but no stable matching as the bottom entry.}
\label{ta:exp}
\end{table}

\section{Conclusion and open questions}
\label{sec:concl}

The most prominent direction for future results is to accelerate our algorithm to reach a fixed parameter tractability result for the discussed case 
or to discover further polynomially solvable instance classes of the popular roommates problem. As argued in Section~\ref{sec:corr}, graphs with a high minimum degree can only be a subject of such investigation if combined with an odd $n$, which is the case we covered in this paper. However, other graph parameters might prove to be fruitful. Parameterized complexity results on other popular matching problems restrict the variability of preferences~\cite{HKMN11,KNN+18} or operate on graphs with a bounded treewidth~\cite{FKPZ19}.

Decomposing a popular matching to a set of edges that is stable on their subgraph and to a set of edges that is popular on their subgraph appears in two further papers. Cseh and Kavitha~\cite{CK21} identified the set of edges that can appear in a popular matching on a bipartite instance by showing that any popular matching $M$ can be decomposed as $M = M_0 \cup M_1$, where $M_0$ is a so-called \emph{dominant} matching in the subgraph induced by the vertices matched in $M_0$, and in the subgraph induced by the remaining vertices, $M_1$ is stable. Kavitha~\cite{Kav19} searched for popular matchings in non-bipartite instances by showing that every popular matching can be partitioned into a stable part and a \emph{truly popular} part. Discovering a connection between these three decompositions might lead to new structural insights.




\bibliography{mybib}
\appendix
\section{Examples}

We now list two examples, one NO-instance and one YES-instance.

\subsection{NO-instance}
We demonstrate a full run of our algorithm on the instance depicted in Figure~\ref{fig:ex}. It is easy to see that the instance admits no stable matching, because each stable matching would have to include the edge $(e,g)$, and after fixing this edge, we are left with a preference cycle of length~3 on vertices $\left\{a,b,d\right\}$, which proves the non-existence of a stable matching.



Since $n$ is odd, each popular matching must leave an odd number of vertices uncovered. In order to check whether there is a popular matching, we now iterate through all vertex sets $U$ with an odd cardinality and no spanned edge, and apply Algorithms~\ref{alg: check U} and~\ref{alg:delete} to determine whether there is a popular matching that leaves exactly the vertices in $U$ uncovered. Table~\ref{ta:ex} summarizes these steps. As the instance admits no popular matching, each tested $P_Z$ fails at some point. The exact explanation on how this happens is deferred to the list marked with capital letters.

\begin{table}[tb]
\centering
 \resizebox{1\textwidth}{!}{  
\begin{tabular}{|l||r|r|r|r|}
\hline
 Case & $P_Z$ & $V(G')$ & $E(G')$ & test  \\ \hline\hline
$U=\{a\}$&  $(f,g),(d,h)$ &  $a,d,f,g,h$ & $(a,d),(d,h),(f,g),(g,h)$ & 1 (\newtag{A}{l1}) \\
  $Z=\{f,g,h\}$   &  $(e,f),(g,h)$  & $a,e,f,g,h$ & $(a,e),(e,f),(e,g),(f,g),(g,h)$ & 1 (\newtag{B}{l2})  \\ \hline
  
    $U=\{b\}$&  $(e,f),(g,h)$ &  $b,e,f,g,h$ & $(e,f),(e,g),(f,g),(g,h)$ & 3 (\newtag{C}{l3}) \\
  $Z=\{e,f,g,h\}$   &  $(a,e),(f,g),(d,h)$  & $V(G)$ & $E(G)$ & 1 (\newtag{D}{l4})  \\ \hline     
  
  $U=\{d\}$&  $(f,g)$ &  $d,f,g$ & $(f,g)$ & 2 (\newtag{E}{l5}) \\
  $Z=\{f,g\}$   &  $(e,f,(g,h)$  & $d,e,f,g,h$ & $(d,e),(d,h),(e,f),(e,g),(f,g),(g,h)$ & 1 (\newtag{F}{l6})  \\ \hline 
    
    $U=\{e\}$&  $(a,b),(d,h)$ &  $a,b,d,e,h$ & $(a,b),(a,d),(a,e),(b,d),(d,e),(d,h)$ & 1 (\ref{l6}) \\
  $Z=\{b,h\}$   &  $(a,b),(g,h)$  & $a,b,e,g,h$ & $(a,b),(a,e),(e,g),(g,h)$ & 1 (\newtag{G}{l7})  \\ 
  &  $(b,d),(g,h)$  & $a,b,e,g,h$ & $(a,b),(a,e),(e,g),(g,h)$ & 1 (\ref{l7})  \\ \hline 
  
    $U=\{f\}$&  $(a,b),(d,h)$ &  $a,b,d,f,h$ & $(a,b),(a,d),(b,d),(d,h)$ & 3 (\newtag{H}{l8}) \\
  $Z=\{a,b,d,h\}$   &  $(a,b),(d,e),(g,h)$  & $V(G)$ & $E(G)$ & 1 (\newtag{I}{l9})  \\ 
  &  $(a,e),(b,d),(g,h)$  & $V(G)$ & $E(G)$ & 1 (\newtag{J}{l10})  \\ \hline 
  
    $U=\{g\}$&  $(a,b),(d,e)$ &  $a,b,d,e,g$ & $(a,b),(a,d),(a,e),(b,d),(d,e),(e,g)$ & 1 (\ref{l7}) \\
  $Z=\{a,b,d\}$   &  $(a,b),(d,h)$  & $a,b,d,g,h$ & $(a,b),(a,d),(b,d),(d,h),(g,h)$ & 1 (\ref{l7})  \\ 
  &  $(a,e),(b,d)$  & $a,b,d,e,g$ & $(a,b),(a,d),(a,e),(b,d),(d,e),(e,g)$ & 1 (\ref{l7})  \\ \hline 
  
    $U=\{h\}$&  $(a,b),(e,f)$ &  $a,b,e,f,h$ & $(a,b),(a,e),(e,f)$ & 2 (\ref{l5}) \\
  $Z=\{a,b,e,f\}$   &  $(a,b),(d,e),(f,g)$  & $V(G)$ & $E(G)$ & 1 (\newtag{K}{l11})  \\ 
  &  $(a,e),(b,d),(f,g)$  & $V(G)$ & $E(G)$  & 1 (\newtag{L}{l12})  \\ \hline
  
  $U=\{a,f,h\}$&  $\varnothing$ &  $a,f,h$ & $\varnothing$ & 3 (\newtag{M}{l13}) \\
  $Z=\varnothing$   &   & &  & \\ \hline
  
  $U=\{b,e,h\}$&  $\varnothing$ &  $b,e,h$ & $\varnothing$ & 3 (\newtag{N}{l14}) \\
  $Z=\varnothing$   &   & &  & \\ \hline
  
  $U=\{b,f,h\}$&  $\varnothing$ &  $b,f,h$ & $\varnothing$ & 3 (\newtag{O}{l15}) \\
  $Z=\varnothing$   &   & &  & \\ \hline
\end{tabular}
}
\caption{The first column contains the tested $U$ set and the set of vertices $Z$, calculated from~$U$. In the second column, we list the possible $P_Z$ matchings. To each of these, the vertices and edges of $G'$ are listed. Finally, we mark which test the current $P_Z$ failed and give a detailed explanation of this failure below the table.}
\label{ta:ex} 
\end{table}

\begin{itemize}
\item[\ref{l1}:] $P_Z$ is not popular in $G'$, because the blocking edge $(g,h)$ can be reached from $a$ in $G'_{P_Z}$.
\item[\ref{l2}:] $P_Z$ is not popular in $G'$, as $(a,e)$ blocks $P_Z$.
\item[\ref{l3}:] Edge $(e,g)$ blocks $P_Z$ and $D=\{f,h\}$. The graph spanned by $V\setminus V'=\{a,d\}$ consists of the edge $(a,d)$ itself, but even this edge is deleted in Algorithm~\ref{alg:delete}. We conclude that after the edge deletions in the third test, no stable matching covers the entire $V\setminus V'$.
\item[\ref{l4}:] $P_Z$ is not popular in $G'$, as $(b,d)$ blocks $P_Z$.
\item[\ref{l5}:] $P_Z$ is stable in $G'$.
\item[\ref{l6}:] $P_Z$ is not popular in $G'$, as $(d,e)$ blocks $P_Z$.
\item[\ref{l7}:] $P_Z$ is not popular in $G'$, as $(e,g)$ blocks $P_Z$.
\item[\ref{l8}:] Edge $(a,d)$ blocks $P_Z$ and $D=\{b,h\}$. The graph spanned by $V\setminus V'=\{e,g\}$ consists of the edge $(a,d)$ itself, but even this edge is deleted in Algorithm~\ref{alg:delete}, because $g\succ_{h} d$. We conclude that after the edge deletions in the third test, no stable matching covers the entire $V\setminus V'$.
\item[\ref{l9}:] $P_Z$ is not popular in $G'$, because the blocking edge $(a,d)$ can be reached from $f$ in $G'_{P_Z}$.
\item[\ref{l10}:] $P_Z$ is not popular in $G'$, because the blocking edge $(a,b)$ can be reached from $f$ in $G'_{P_Z}$.
\item[\ref{l11}:] $P_Z$ is not popular in $G'$, because the blocking edge $(e,g)$ can be reached from $h$ in $G'_{P_Z}$.
\item[\ref{l12}:] $P_Z$ is not popular in $G'$, because the blocking edge $(a,b)$ can be reached from $h$ in $G'_{P_Z}$.
\item[\ref{l13}:] In Algorithm~\ref{alg:delete} we delete $(b,d)$ since $a\succ_{b} d$. After this deletion, no matching covers $V\setminus V'=\{b,d,e,g\}$. 
\item[\ref{l14}:] In Algorithm~\ref{alg:delete} we delete $(a,d)$ since $b\succ_{d} a$. After this deletion, no matching covers $V\setminus V'=\{a,d,f,g\}$. 
\item[\ref{l15}:] In Algorithm~\ref{alg:delete} we delete $(a,d)$ since $b\succ_{d} a$. After this deletion, no matching covers $V\setminus V'=\{a,d,e,g\}$.
\end{itemize}

\subsection{YES-instance}
\begin{figure}[htb]
	\centering
		\begin{minipage}{0.2\textwidth}
		\[
		\begin{array}{llllll}
		a  : e, \ & b \\
		b  : e, \ & h, \ & g, \ & d, \ & a, \ & f \\
		d  : h, \ & e, \ & b \\
		e  : d, \ & h, \ & b, \ & g, \ & a, \ & f \\
		f  : b, \ & e \\
		g  : e, \ & b, \ & h \\
		h  : g, \ & e, \ & d, \ & b
		\end{array}
		\]
	\end{minipage}\hspace{30mm}\begin{minipage}{0.35\textwidth}	
		\begin{tikzpicture}[scale=0.97, transform shape]
		\pgfmathsetmacro{\d}{2}
		\pgfmathsetmacro{\b}{3}
		
    \node[vertex, label=below:$a$] (A1) at (0,0) {};
	\node[vertex, label=below:$b$] (A2) at ($(A1) + (1.5*\d, 0)$) {};
	\node[vertex, label=right:$d$] (A3) at ($(A2) + (\d, \d)$) {};
	\node[vertex, label=left:$h$] (A7) at ($(A1) + (-0.5*\d, \d)$) {};
	\node[vertex, label=above:$g$] (A6)
	at ($(A7) + (0.5*\d, \d)$) {};
	\node[vertex, label=left:$f$] (A5)
	at ($(A6) + (\d, 0.5*\d)$) {};
	\node[vertex, label=above:$e$] (A4)
	at ($(A3) + (0, \d)$) {};
	
	\draw [thick] (A1) -- node[edgelabel, near start] {2} node[edgelabel, near end] {5} (A2);
	\draw [thick] (A2) -- node[edgelabel, near start] {4} node[edgelabel, near end] {3} (A3);
	\draw [thick] (A3) -- node[edgelabel, near start] {2} node[edgelabel, near end] {1} (A4);
	\draw [thick] (A1) -- node[edgelabel, near start] {1} node[edgelabel, near end] {5} (A4);
	\draw [thick] (A2) -- node[edgelabel, near start] {1} node[edgelabel, near end] {3} (A4);
	\draw [thick] (A2) -- node[edgelabel, near start] {6} node[edgelabel, near end] {1} (A5);
	\draw [thick] (A2) -- node[edgelabel, near start] {3} node[edgelabel, near end] {2} (A6);
	\draw [thick] (A2) -- node[edgelabel, near start] {2} node[edgelabel, near end] {4} (A7);
	\draw [thick] (A3) -- node[edgelabel, near start] {1} node[edgelabel, near end] {3} (A7);
	\draw [thick] (A4) -- node[edgelabel, near start] {6} node[edgelabel, near end] {2} (A5);
	\draw [thick] (A4) -- node[edgelabel, near start] {4} node[edgelabel, near end] {1} (A6);
	\draw [thick] (A4) -- node[edgelabel, near start] {2} node[edgelabel, near end] {2} (A7);
	\draw [thick] (A6) -- node[edgelabel, near start] {3} node[edgelabel, near end] {1} (A7);
\end{tikzpicture}
\end{minipage}
\caption{Example instance on 7 vertices that admits a popular matching but no stable matching.}
\label{fig:ex2}
\end{figure}
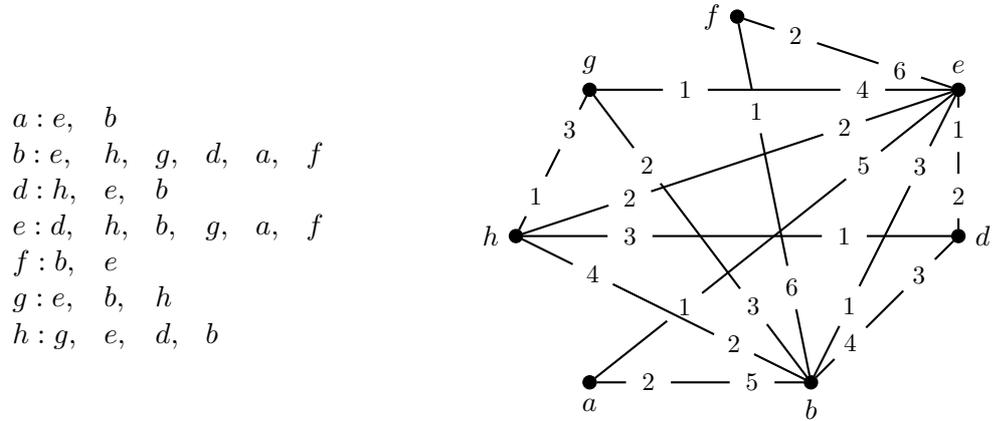
The instance depicted in Figure~\ref{fig:ex2} admits a popular matching, but no stable matching. Since $n$ is odd, each popular matching must leave an odd number of vertices uncovered. In order to check whether there is a popular matching, we start iterating through all vertex sets $U$ with an odd cardinality and no spanned edge, and apply Algorithms~\ref{alg: check U} and~\ref{alg:delete} to determine whether there is a popular matching that leaves exactly the vertices in $U$ uncovered. Table~\ref{ta:ex2} summarizes these steps. A popular matching consisting of edges $(a,b),(d,h),(e,g)$ is found in the line marked by~(*). The algorithm then terminates with outputting this matching.



\begin{table}[bt]
\centering
 \resizebox{1\textwidth}{!}{
\begin{tabular}{|l||r|r|r|r|r|}
\hline
 Case & $P_Z$ & $V(G')$ & $E(G')$ & test & reason for test failure\\ \hline\hline
$U=\{a\}$ &  $(b,d),(e,f),(g,h)$ & $V(G)$ & $E(G)$ & 1 & $(a,e)$ blocks $P_Z$\\
  $Z=\{d,f,g,h\}$   &  $(b,f),(d,e),(g,h)$  & $V(G)$  & $E(G)$ & 1 & $(a,e)$ blocks $P_Z$\\ 
   & $(b,f),(d,h),(e,g)$ & $V(G)$ & $E(G)$ & 1 & $(a,b)$ blocks $P_Z$\\
   & $(b,g),(d,h),(e,f)$ & $V(G)$ & $E(G)$ & 1 & $(a,b)$ blocks $P_Z$\\ \hline
$U=\{b\}$ & $\varnothing$ & $b$ & $\varnothing$ & 3 & no complete matching in $G\setminus G'$\\
$Z=\varnothing$ & & & & & \\ \hline
$U=\{d\}$ & $(a,b),(e,f),(g,h)$ & $V(G)$ & $E(G)$ & 1 & $(d,e)$ blocks $P_Z$\\
$Z=\{a,f,g\}$ & $(a,e),(b,f),(g,h)$ & $V(G)$ & $E(G)$ & 1 & $(d,e)$ blocks $P_Z$\\ \hline
$U=\{e\}$ & $\varnothing$ & $e$ & $\varnothing$ & 3 & no complete matching in $G\setminus G'$\\
$Z=\varnothing$ & & & & & \\ \hline
$U=\{f\}$ & $(a,b),(d,e),(g,h)$ & $V(G)$ & $E(G)$ & 1 & the blocking edge $(b,g)$ is reachable\\ 
$Z=\{a,d,g,h\}$ & & & & & from $f$ in $G_{P_Z}$ via an alternating path\\
 & $(a,e),(b,d),(g,h)$ & $V(G)$ & $E(G)$ & 1 & the blocking edge $(d,e)$ is reachable\\
 & & & & & from $f$ in $G_{P_Z}$ via an alternating path\\
 & $(a,e),(b,g),(d,h)$ & $V(G)$ & $E(G)$ & 1 & the blocking edge $(e,g)$ is reachable\\
 & & & & & from $f$ in $G_{P_Z}$ via an alternating path\\ 
  & $(a,b),(d,h),(e,g)$ & $V(G)$ & $E(G)$ & - & *\\\hline
                        
\end{tabular}
}
\caption{The first column contains the tested $U$ set and the set of vertices $Z$, calculated from~$U$. In the second column, we list the possible $P_Z$ matchings. To each of these, the vertices and edges of $G'$ are listed. Then we mark which test the current $P_Z$ failed and explain this failure in the last column.}
\label{ta:ex2}
\end{table}


\end{document}